\documentclass[conference, letterpaper, 10pt, final, comsoc]{IEEEtran}

\IEEEoverridecommandlockouts
\usepackage[left=1.62cm,right=1.62cm,top=1.9cm]{geometry}
\usepackage{amsmath,amssymb,amsfonts,amsthm}
\usepackage{graphicx}
\usepackage{bbm}
\usepackage{xcolor}
\newcounter{thm}
\newtheorem{theorem}[thm]{Theorem}

\newtheorem{lemma}[thm]{Lemma}

\newcommand*{\rom}[1]{\uppercase\expandafter{\romannumeral #1\relax}}

\newcommand{\eps}{{\varepsilon}} 
\newcommand{\ket}[1]{{|#1\rangle}}   
\newcommand{\bra}[1]{{\langle#1|}}   
   
\newcommand{\erf}{{\mathrm{erf}}}
\newcommand{\tr}{{\mathrm{tr}}}

\allowdisplaybreaks

\makeatletter

\counterwithout{equation}{section} 

\title{Performance of Quantum Preprocessing under Phase Noise}

\author{
   \IEEEauthorblockN{Zuhra Amiri\IEEEauthorrefmark{1}, Boulat A. Bash\IEEEauthorrefmark{2}, Janis N\"otzel\IEEEauthorrefmark{1}}
   \IEEEauthorblockA{
       \IEEEauthorrefmark{1}Emmy-Noether Group Theoretical Quantum Systems Design,  Lehrstuhl f\"ur Theoretische Informationstechnik,\\ Technische Universit\"at M\"unchen\\ \IEEEauthorrefmark{2}Electrical \& Computer Engineering, Optical Sciences, University of Arizona, Tucson, Arizona, USA\\
  \{zuhra.amiri, janis.noetzel\}@tum.de, boulat@arizona.edu 
    }
    \thanks{\copyright 2022 IEEE. Personal use of this material is permitted. Permission from IEEE must be obtained for all other uses, in any current or future media, including reprinting/republishing this material for advertising or promotional purposes, creating new collective works, for resale or redistribution to servers or lists, or reuse of any copyrighted component of this work in other works.}
}

\begin{document}

\maketitle
\begin{abstract}
    Optical fiber transmission systems form the backbone of today's communication networks and will be of high importance for future networks as well. Among the prominent noise effects in optical fiber is phase noise, which is induced by the Kerr effect. This effect limits the data transmission capacity of these networks and incurs high processing load on the receiver. At the same time, quantum information processing techniques offer more efficient solutions but are believed to be inefficient in terms of size, power consumption and resistance to noise. Here we investigate the concept of an all-optical joint detection receiver. We show how it contributes to enabling higher baud-rates for optical transmission systems when used as a pre-processor, even under high levels of noise induced by the Kerr effect.
\end{abstract}

\begin{section}{Introduction}
    The present and the future of our societies relies more and more on high speed connectivity between a growing number of services. Advanced communication systems support the operation of everything from communication between machines to communication between humans. Real-time video services as well as emerging applications like telemedicine and connected cars will further increase the demand for connectivity, which can be anticipated from the Cisco annual internet report \cite{cisco}. At the same time, increasingly smaller wireless cells will put a tremendous load on the optical fiber backbone, where not only increasing data rates but also reduced latency and lower power consumption are demanded at the same time \cite{6Grequirements}. 
    
    In recent works \cite{6Gquantum} it has thus been speculated that quantum information processing (QIP) technologies should play a larger role in the development of the next generation of mobile networks. Quantum primitives such as squeezed light and entanglement at transmitter and receiver can significantly improve communication rates. However, the drawback of QIP is that many solutions only exist on paper, with many mathematical constructs lacking their physical counterpart, and existing ones being bulky and requiring excessive cooling to combat environmental noise, which make them hard to use in practice.
    
    In this work, we investigate a QIP technique, namely the concept of the joint detection receiver (JDR), which promises a practical realization. In the same way as Shannon's work led engineers to perform error correction over multiple received bits to reach superior performance, the work \cite{Holevo1998c,Giovannetti2015} of Holevo and follow-up research motivates to perform error correction over multiple received \emph{pulses}. This latter operation will in practice be carried out by a JDR. Such JDR can be thought of as operating fully in the optical domain, ultimately producing a bit which is handed over to a higher layer for processing. In \cite{guha2011structured} the first proposal for the design of a JDR has been made. Due to its superior performance in the low photon number regime, this device has previously be seen as an optimal choice for deep space communications. While this is certainly true, in the recent work \cite{noetzelrosati2022} it has been pointed out that power-limited communication under high baud-rates also inevitably leads to a low number of received photons per pulse. Observed trends for baud-rates \cite{highBaudrateComms} let us conjecture that future systems ten to twenty years from now will likely operate at baud-rates in the area from $300$GBd to well above $400$GBd. As techniques emerge which even allow the conversion of entire frequency domains (e.g. C- to O-band) \cite{cbandConverter}, there does not seem to be a natural limit for increasing baud-rates, but rather technological hurdles to be overcome. However, when increasing baud-rates under a power limit, current systems fall short of realizing any reasonable gain \cite{noetzelrosati2022}. In sharp contrast, systems utilizing the QIP technique of joint detection can be expected to benefit from the observed trend \cite{highBaudrateComms}.\\
    Data transmission techniques utilizing optical fiber need to deal with fiber nonlinearities. In this domain, the recent work \cite{ludovicoWilde} has discovered corresponding capacities. As any realization of the QIP potential will have to be based on practical design, we study here the performance limits of a practically implementable design based on phase shift keying rather than information-theoretic bounds. Under any fixed power limit, high baud-rates will eventually induce the low photon numbers where the JDR technology outperforms its classical analog, and in such domain it is inevitably at some point optimal to use, among all phase shift keying formats, the one with the lowest modulation order. Thus in this work, we study binary phase shift keying (BPSK).\\
    Our model for phase noise is derived from the Kerr effect, which results in so-called phase noise and has been the subject of investigation for example in the recent work \cite{kunzParisBanaszekKerrMedium}. 
    In our analysis, we utilize this model in simulations and in the derivation of an analytical channel model. 
\end{section}

\begin{section}{System Model and Notation}
    Consider a channel with complex valued input and output 
    \begin{align}
        Y = \tau\cdot X\cdot e^{\mathbbm{i}\Phi}, X\in\mathbbm C, \phi \in [0, 2\pi)
    \end{align}
    where $X$ is the signal, $\tau=e^{-a\cdot L}\in[0,1]$ the transmittivity, $a$ the attenuation coefficient, $L$ the fiber length in $km$, and $e^{-\mathbbm{i}\Phi}$ a random phase noise term \cite{kramerSystemModel}.
    When investigating the potential of a JDR for such system, it is sufficient to replace $X$ by a coherent state $|X\rangle$ which is an element of the Fock space $\mathcal F$ \cite{banaszekQuantumLimits}.     The sender encodes their message into coherent states $|\alpha\rangle=\exp(-|\alpha|^2/2)\sum_{n=0}^\infty\alpha^n/\sqrt{n!}|n\rangle$, where $\{\ket{n}\}_{n\in\mathbb N}$ is the photon number basis of $\mathcal F$. For binary phase shift keying (BPSK), the set of signals states is $S_1=\{\ket{\alpha},\ket{-\alpha}\}$. The signal energy per pulse is given by $E= \hbar \omega_0 |\alpha|^2$, with $\hbar \omega_0$ being the energy of a single photon at carrier frequency $\omega_0$. To use the joint detection concept discovered in \cite{guha2011structured} for BPSK signals one can construct a code book of $2\cdot n$ signals $\ket{v_k(s\alpha)}$, $k=1,\ldots,n$, $s=\pm1$ with code-words taking the form 
    \begin{equation}
        \ket{v_k(\alpha)} = \otimes_{j=0}^{n-1} \ket{(H_n)_{j,k} \alpha},
    \end{equation}
    where the symmetric matrix $H_n$ is defined as
    \begin{equation}
        (H_n)_{j,k} = (-1)^{j\cdot k},\quad j \cdot k = \textstyle\sum_{t=0}^{\log_2 n} j_t k_t,
    \end{equation}
    with $j \cdot k$ being the bitwise scalar product of the binary representations of $j,k = 0,...,n-1$ \cite{rosati2016}. The receiver employs a Hadamard transformation $\hat{U}_{Had}^{(n)}$, resulting in the state
    \begin{align}
        \hat{U}_{Had}^{(n)}\ket{v_k(\alpha)} = \ket{w_k(\alpha)} = \ket{\sqrt{n} \alpha}_k \left( \otimes_{j\neq k} \ket{0}_j \right).
    \end{align}
    As can be seen, all but one output mode of the Hadamard receiver are in the vacuum state, and the desired phase can be decoded from the $k$-th output mode in state $\ket{\sqrt{n}\alpha}$. In this work, we employ homodyne detection for this task. $\hat{U}_{Had}^{(n)}$
    can be implemented using an array of beamsplitters. Each beamsplitter $U_{BS}$ transforms an incoming pair of coherent states as 
    \begin{align}\label{eqn:beamsplitter}
        U_{BS}\ket{\alpha}\otimes\ket{\beta} = \ket{(\alpha+\beta)/\sqrt{2}}\otimes\ket{(\alpha-\beta)/\sqrt{2}}.
    \end{align}
    The receiver's performance can be evaluated by using classical statistical tools applied to the complex numbers $\alpha, \beta$. To account for the Kerr effect, we consider phase noise. This leads to states of the form $\ket{\alpha e^{\mathbbm{i}\phi}}$ with $\phi$ a random phase. 
    For a sequence $\Phi^n:=(\phi_1,\ldots,\phi_n)$ of such phases, the output of the Hadamard receiver of order $n=2^K$ at port $k'$ given input $\nu_k(\alpha)$ is
    \begin{equation}
      \ket{\Lambda^n_{k,k'}(\alpha)}:=\ket{2^{-K/2}\alpha\left(\sum_{m=1}^n H_{k,m}H_{m,k'} e^{\mathbbm{i} \phi_m}\right)}.
\end{equation}
In order to quantify the distribution of the phases $\phi$, we consider the phase noise model derived in  \cite{kunzParisBanaszekKerrMedium}, which lets $\phi$ be distributed according to a normal distribution with a variance that is with our parameter choices estimated as 
\begin{align}
    \sigma^2\approx 6 \cdot 10^{-19}\cdot b,\label{eqn:phase-noise}
\end{align}
where $b$ is the baud-rate (see Section \ref{subsection:noise model}). Each of the homodyne detectors at the output ports of the receiver operates based on a threshold $\eps>0$ as follows: If $m$ is the measurement result of the detector, then the received signal will be set to $\alpha$ if $m>\eps$, $0$ if $m\in(-\eps,\eps)$ and $-\alpha$, else. This yields a statistical input-output relation (see Subsection \ref{subsec:detection} for details)
\begin{align}\label{eqn:conditional-distribution}
    q(y^n|\alpha,k)=p_C^\epsilon(y_k|\alpha)\prod_{k'\neq k}p_I^\epsilon(y_{k'}|\alpha)
\end{align}
where $y^n\in\{-\alpha,0,\alpha\}$ are detection results at the $n=2^K$ output ports of the receiver. $p_C$ and $p_I$ denote the correct and incorrect detection probabilities respectively.
The capacity of the system without the JDR is calculated as $C(b,E):=b\cdot\max_{p_A}I(Y;A)$ where $p_A$ is the distribution of the phases, and as 
\begin{align}\label{def:hadamard-capacity}
    C(b,E,n):=b\cdot I(Y^n;AN)/n
\end{align}
with $A$ and $N$ uniformly distributed over $\{-\alpha,\alpha\}$ and $\{1,\ldots,n\}$ if quantum pre-processing, JDR, is used.
\end{section}

\begin{section}{Results}
To analyze the statistical properties of the received signal we set  $t(i,k):=(H_{i,m}H_{k,m})_{m=1}^n$. It holds $t(i,k)_m\in\{-1,1\}$ and $\sum_mt(i,k)_m=\delta(i,k)$. The received signal then reads as
\begin{align}
    \ket{\Lambda^n_{k,k'}(\alpha)}
    = \begin{cases}
    \ket{2^{-K/2}\alpha \sum_{m}e^{\mathbbm{i}\phi_m} }&,k=k'\\
    \ket{2^{-K/2}\alpha \sum_{m}t(k,k')_me^{\mathbbm{i}\phi_m}  }&,\mathrm{else}\label{eqn:received-signal}
    \end{cases}.
\end{align}
Equation \eqref{eqn:received-signal} naturally allows us to state the following:
\begin{theorem}\label{thm:result}
    Let for each $K\in\mathbb N$ $n:=2^K$ and let the random variables $\Phi_1,\ldots,\Phi_n$ be i.i.d according to a measure $\mu$ on $[0,2\pi)$ and $t>0$. For all $k,k'=1,\ldots,n$:
    \begin{align}
        \mathbb P\bigg(\bigg|\Lambda^n_{k,k'}(\alpha) - \sqrt{n}\alpha\delta(k,k')\mathbb E_\mu [e^{\mathbbm{i}\Phi}]\bigg|\geq \tfrac{t|\alpha|}{\sqrt{n}}\bigg)&\leq 4e^{-t^2/n}.
    \end{align}
\end{theorem}
We conclude the Hadamard receiver asymptotically transforms any signal $\alpha H_{1k},\ldots\alpha H_{nk}$ into output $\approx\sqrt{n}\alpha\mathbb E_\mu[e^{\mathbbm{i}\Phi}]$ at output port $k$, whereas at output port $k'\neq k$ one receives almost no signal. To optimize the overall system performance we approach the problem of optimizing the homodyne detectors at the output ports of the system as follows: 
\begin{align}\label{eqn:mathematical-problem-statement}
    \eps_{\max}:=\arg\max_{\eps}\ & p_C^{\epsilon}(\alpha|\alpha)\prod_{k\neq k'}p_I^{\epsilon}(0|0).
\end{align}
This approach is motivated from three observations: {\bf{ 1.}} The calculation of $\epsilon$ in \eqref{eqn:conditional-distribution} is becoming more challenging as $n$ grows. {\bf 2.} As $n$ grows, the central limit theorem predicts that $\Lambda^{n}_{k,k'}(\alpha)$ will be approximately a normal distribution. {\bf 3.} the variance $\sigma$ of said normal distribution can be efficiently computed. Turning the observation {\bf 2.} into a quantitative \emph{assumption}, we can prove that setting the detection parameters according to the solution of \eqref{eqn:mathematical-problem-statement} yields near-optimal results once the solution $\eps_{\max}$ yields $p_C^{\epsilon_{\max}}(\alpha|\alpha)\prod_{k\neq k'}p_I^{\epsilon_{\max}}(0|0)\approx1$:
\begin{theorem}\label{thm:capacity-bound}
For all $b,E,n$ we have $C(b,E,n)\leq b\cdot\log(n)/n$. If $p_I^\epsilon(y|\pm\alpha)=d(y)$ for all $y\in\{-\alpha,0,\alpha\}$ and  $p_C^\epsilon(\alpha|\alpha)\geq1-\delta$ and $d(0)\geq1-\delta$ and $\delta<n^{-1}$, then 
\begin{align}
    C(b,E,n)\geq b\left((1-\delta)\log n - 5\cdot h(\delta)\right)/n,
\end{align}
where $h$ is the binary entropy.
\end{theorem}
We utilize problem statement \eqref{eqn:mathematical-problem-statement} to obtain numerically values for setting the detection threshold $\epsilon$. Thereby we obtained lower bounds for $C(b,E,n)$ for different values of $b$. In this process, we replaced the Gaussian distribution of $\Phi$ with a von Mises distribution (see Subsection \ref{subsection:noise model} for details and figures \ref{fig:mutualinfo} and \ref{fig:capacities} for numerical results concerning our particular application). Using the von Mises distribution reduced the dependence of Eq. \eqref{eqn:mathematical-problem-statement} on the Hadamard receiver size, as explained in Subsubsection \ref{subsubsec:detector-optimization-procedure}.

Our simulation results show the different capacities with increasing baud-rates and phase noise: 
\begin{figure}[h]
    \includegraphics[width=.48\textwidth]{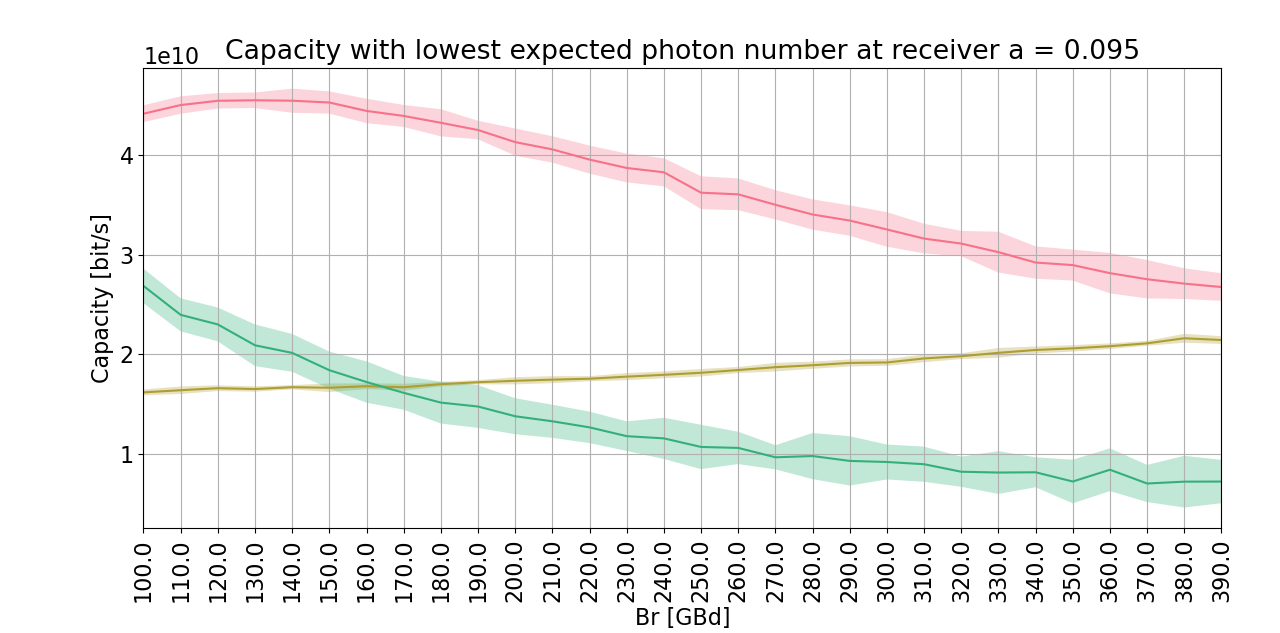}
    \caption{Average capacity plotted over the baud-rate. The green line represents the classical Shannon capacity, the red line the capacity of the Hadamard receiver with order $n=4$ and the brown line with order $n=32$. At $130Gbd$, the received photon number is approximately $0.29$ photons per pulse. The shaded regions are error bars derived from the empirical variance of simulation results. We used the attenuation coefficient $ a = 0.046$ and fiber length $L=250km$ resulting in a transmittivity of $\tau\approx10^{-5}$.}
    \label{fig:capacitybaudrates}
\end{figure}
As the performance of the JDR and of the standard homodyne receiver are both hard to evaluate analytically for specific parameters, the performance was simulated \cite{ourCode}. For the attenuation we used the formula $\tau=\exp{-0.046\cdot L}$ modelling approximately optical fibers \cite{kunzParisBanaszekKerrMedium} of length $L$ kilometers.  The attenuation coefficient $a = 0.046$ was taken from \cite[Table 1]{kunzParisBanaszekKerrMedium}. It models attenuation in SMF-28 fiber for communication in the C-band at $1550$nm, which is the most widely used system choice in communication networks.

\begin{figure}
    \includegraphics[width=.48\textwidth]{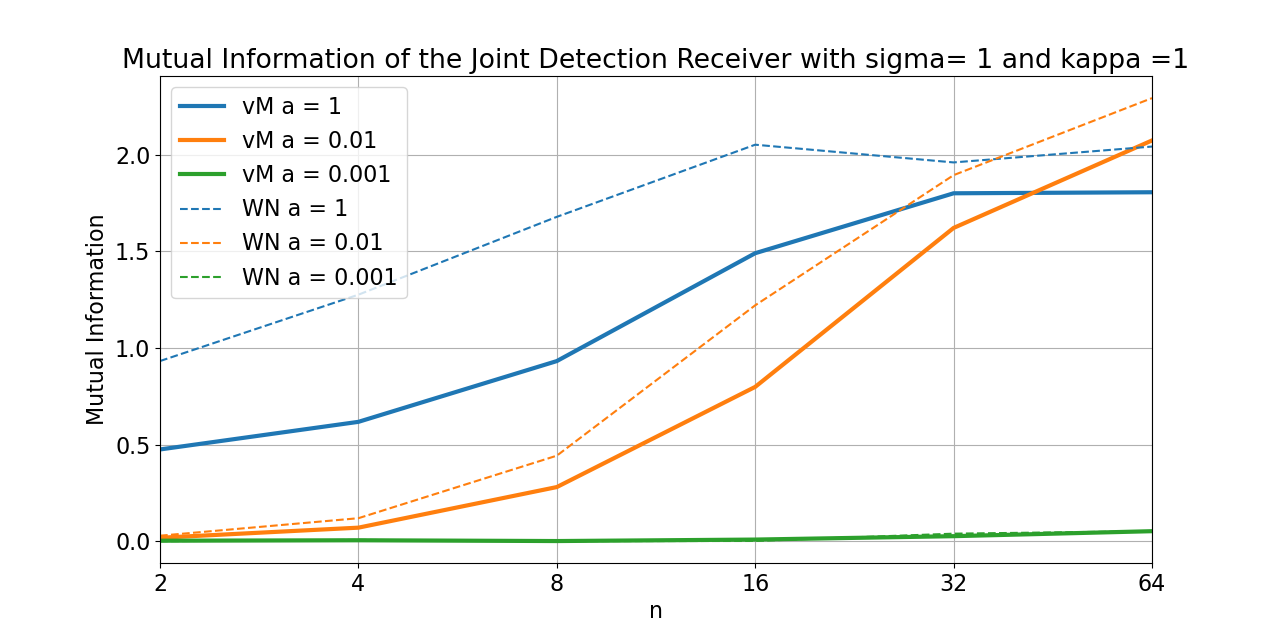}
    \caption{In this Figure the mutual information of the Hadamard receiver is plotted over the Hadamard order $n$. The dashed lines represent the capacities with the wrapped normal distribution (WN) and the solid lines are representing the von Mises(vM) distribution chosen as the phase noise distribution.}
    \label{fig:mutualinfo}
\end{figure}
\begin{figure}
    \includegraphics[width=.48\textwidth]{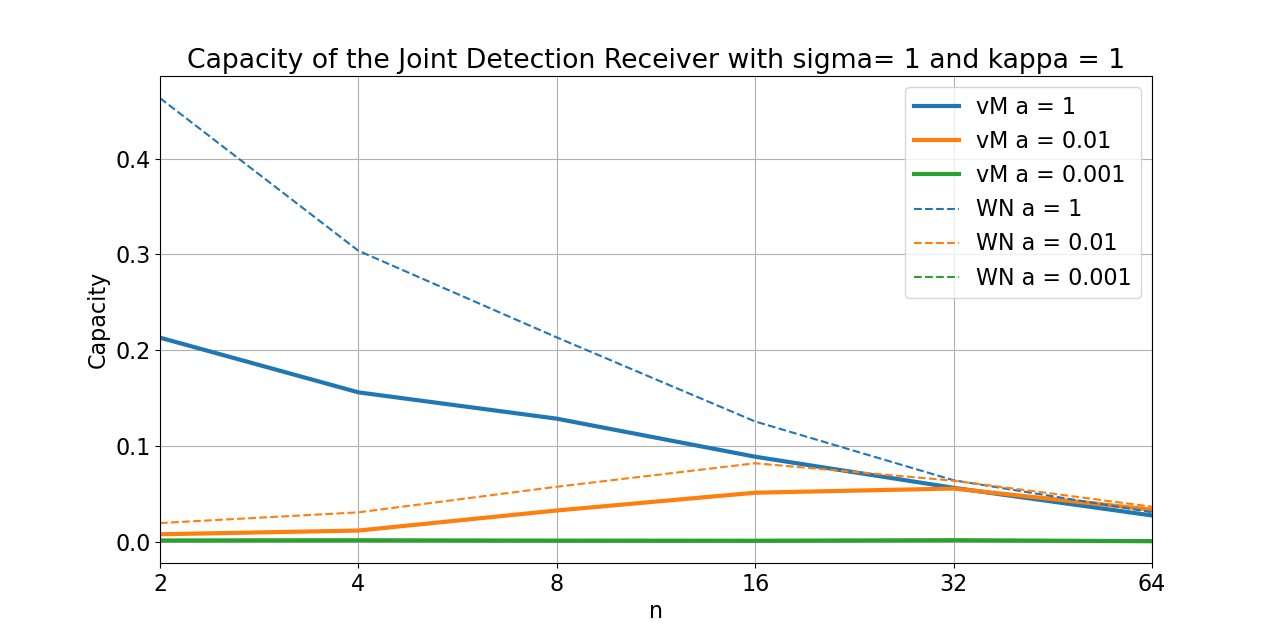}
    \caption{In this Figure the capacity of the Hadamard receiver is plotted over the Hadamard order $n$. The dashed lines represent the capacities with the wrapped normal distribution (WN) and the solid lines are representing the von Mises (vM) distribution chosen as the phase noise distribution.}
    \label{fig:capacities}
\end{figure}
    
Figure \ref{fig:capacitybaudrates} displays simulation results for $L=250km$ it can be seen that with quantum pre-processing, we can achieve a higher capacity at high baud-rates than with an only classical system. From Figure \ref{fig:capacitybaudrates} we can also see that higher Hadamard orders ($n=32$) tend to be more beneficial when the ratio of $E/b$ is extremely small, while lower orders $(n=4)$ can already bring performance improvements. The maximum of the red curve at around $130$Gbd restates the fact that JDRs outperform conventional receivers in the low photon regime where $\tau\cdot E/b$ is small ($\ll 1$) for our specific design and noise model. 
\section{Discussion}
\subsubsection{Advantage of the Hadamard receiver}
While it was known \cite{rosati2016, guha2011structured} that the Hadamard receiver outperforms a conventional receiver, our work clarifies that such advantage persists given a scaling of phase noise with baud-rate in a situation where the Hadamard receiver is used only as a pre-processing stage in the design. Note here that we ignore further non-idealities, such as the common mode rejection ratio in the homodyne receiver, or the estimation of global phases.
\subsubsection{Accuracy of Approximation with Van Mises Distribution}
From figures \ref{fig:mutualinfo} and \ref{fig:capacities} one can see that the results obtained by using the von Mises distribution get closer to those obtained based on the wrapped normal distribution, if $n$ grows. 
\subsubsection{Asymptotics for Problem \eqref{eqn:mathematical-problem-statement}}
By setting $t = c\cdot n^{3/4}$ (with suitable choice of $c$) in Theorem \ref{thm:result} one can deduce that, for every $\delta\in(0,1)$, problem statement \eqref{eqn:mathematical-problem-statement} yields a solution $\eps_{max}\leq\delta$: Namely,
    \begin{align*}
        \mathbb P(|\Lambda^n_{k,k'}(\alpha) - \sqrt{n}\alpha\delta(k,k')\mathbb E_\mu [e^{\mathbbm{i}\Phi}]|\geq c|\alpha|n^{1/4})&\leq 4e^{-\sqrt{n} c^2}
    \end{align*}
so that for $c=\mathbb E_\mu[e^{\mathbbm{i}\Phi}]$ and every choice $(\alpha,k)$ of the transmitter the received signal at port $k$ is concentrated in the interval between $\alpha\sqrt{n}\mathbb E_\mu[e^{\mathbbm{i}\Phi}](1\pm n^{-1/4})$ while the received signal at every other port is in another interval between numbers $\alpha\sqrt{n}\mathbb E_\mu[e^{\mathbbm{i}\Phi}](0\pm n^{-1/4})$.
Based on the detection probabilities listed in Subsection \ref{subsec:detection}, we see that for every $\delta>0$ and large enough $n$ a threshold $\eps_n:=\alpha\sqrt{n}\mathbb E_\mu[e^{\mathbbm{i}\Phi}]/2$ is asymptotically good enough to simultaneously achieve $d(0)\geq1-\delta$ and $p_C^\eps(\alpha|\alpha)\geq1-\delta$.
\subsubsection{Detector Optimization Procedure\label{subsubsec:detector-optimization-procedure}}
Finding the optimal solution to \eqref{eqn:mathematical-problem-statement} for arbitrary distributions is a challenging task. 
For the simulations, we thus used a semi-heuristic approach which makes use of the Gaussian shape $\Lambda^n(\alpha)$ when $n$ grows. The basic idea is as follows: For an empirical mean $\bar{X}_n = 1/n \sum_n X_i$, with $X_i$ being the phases $e^{\mathbbm{i}\Phi_i}$, with mean $\mathbb E_\mu[e^{\mathbbm{i}\Phi}]$ and variance $\sigma^2$ of $X_i$, we know from the central limit theorem that, asymptotically,  $\bar{X}_n$ will be distributed as $N(\mathbb E_\mu[e^{\mathbbm{i}\Phi}], \sigma^2/n)$. Since we have $\Lambda_{k,k}^n(\alpha)=\alpha\sqrt{n}\bar{X}_n$ we thus know that (asymptotically) $\Lambda_{k,k}^n(\alpha)$ will be distributed as $N(\alpha\sqrt{n}\mathbb E_\mu[e^{\mathbbm{i}\Phi}], |\alpha|^2\sigma^2)$,
and similarly we get for $k\neq k'$ that $\Lambda_{k,k'}^n(\alpha)$ will be distributed according to $N(0, 2|\alpha|^2\sigma^2)$.

Since the expectation value and variance of the von Mises distribution with zero mean is given by the first and second raw moment respectively (see Lemma \ref{lem:raw moments} and Equation \eqref{def:kappa})
\begin{align}
    \mathbb E_\mu[e^{\mathbbm{i}\Phi}] &= \frac{I_1(\kappa)}{I_0(\kappa)}\\
      \sigma^2_{vM} &= \left(\frac{I_2(\kappa)}{I_0(\kappa)}-\frac{I_1(\kappa)^2}{2I_0(\kappa)^2}\right) - \frac{I_1(\kappa)^2}{2I_0(\kappa)^2},
  \end{align}
we can find, for every $n=2^K$, the optimal $\eps$ in \eqref{eqn:mathematical-problem-statement} by taking into account the detection probability \eqref{eq:detectalpha} as well as the assumed Gaussian shape \begin{align}
    f_t(\beta|\alpha):=\tfrac{1}{t\sigma_{vM}\sqrt{2\pi}}\exp{-\left(\tfrac{\beta-\sqrt{n}E_\mu[e^{\mathbbm{i}\Phi}]\alpha}{2t\sigma_{vM}}\right)^2}
\end{align} 
of the received signal, to arrive at
\begin{align}
    p_C^\eps(\alpha|\alpha)=\int_{\mathbb R}f_1(\beta|\alpha) \frac{1}{2}(1-\erf(\sqrt{2}(\beta-\eps)))d\beta
\end{align}
\begin{align}
    d(0)=\int_{\mathbb R} \frac{f_2(\beta|0)}{2}\sum_{x=0}^1\erf(\sqrt{2}(\eps +(-1)^x\beta))d\beta,
\end{align}
with $\erf(a)$ being the error function
    \begin{align}
        \erf(a) = \frac{2}{\sqrt{\pi}} \int_0^a d\xi e^{-\xi^2}.
    \end{align}
As the computation of $p_C^\eps(\alpha|\alpha)d(0)^{n-1}$ does not depend critically on $n$ anymore, the calculation of $\eps_{\max}$ is efficient even for large values of $n$. This is important for the regime where the ratio $\alpha\sqrt{n}/\sigma$ becomes low. When baud-rate and noise level are low, $\eps=\alpha\sqrt{n}\mathbb E_\mu[e^{\mathbbm{i}\Phi}]/2$ will be a sufficient choice.

\end{section}

\begin{section}{Methods and Proofs}

    \subsection{Noise Model}\label{subsection:noise model}
    Setting $\xi:=\gamma\hbar\omega_0/(a\cdot b)$ where $\gamma=1$ is a nonlinear interaction coefficient (discussed below \cite[Equation (7)]{kunzParisBanaszekKerrMedium}) and $a$ the attenuation parameter we let $N$ be the transmitted number of photons per second so that we get \cite[Eq. (11)]{kunzParisBanaszekKerrMedium} 
    \begin{align}
        \sigma^2=4\cdot\xi^2\cdot N \cdot b^{-1}\cdot\left(2-\tau-\tau(1-\log\tau)^2\right),
    \end{align}
    with $\tau\in(0,1)$ being the transmittivity.
    Assuming transmission at $1550$nm over a standard SMF-28 fiber with attenuation parameter $a=0.046$ \cite[Table 1]{kunzParisBanaszekKerrMedium} we arrive  at
    \begin{align}
        \xi=2.8\cdot 10^{-18}\cdot b,
    \end{align}
    and thus since $\tau=e^{-a\cdot 250}\ll1$ we get
    \begin{align}
        \sigma^2\approx6\cdot 10^{-35}\cdot b\cdot N.
    \end{align}
    For $\approx1mW$ transmit power we use $N=10^{16}$ and therefore
    \begin{align}
        \sigma^2\approx6\cdot 10^{-19}\cdot b.
    \end{align}
    In order to recover the dependence of the phase noise on the signal energy we note that $1/\sigma^2$ plays a role similar to that of $\kappa$ in the von Mises distribution, so that we set
    \begin{align}\label{def:kappa}
        \kappa=10^{19}\cdot b^{-1}/6.
    \end{align}
    As has been noted in \cite{colletLewisDiscriminatingWNDandVMD} the difference between the two models is negligible for sample sizes below $200$. For large sample sizes it is visible for values of $\kappa$ in the interval $(0.1,10)$.

    \subsection{Detection}\label{subsec:detection}
    We use the homodyne detector at each output port of the receiver, with the following POVMs
    \begin{align}
        \Pi_x = \ket{x}\bra{x},
    \end{align}
    with 
    \begin{align}
        \ket{x} =  \left(2/\pi\right)^{\frac{1}{4}} e^{-x^{2}}\textstyle{\sum_{n=0}^{\infty}} H_{n}(\sqrt{2}x)/\sqrt{2^{n}n!}\ket{n},
    \end{align}
    with $H_n(x)$ being Hermite polynomials.
    Then we have the outcome $x$ of the homodyne detector as
    \begin{align}
        p(x|\beta) = \tr\{ \ket{\beta}\bra{\beta}\Pi_x\} = \sqrt{2/\pi}\exp(-2(x-\beta))^2,
    \end{align}
    with $\beta$ being a coherent state. In the case of zero noise $\beta$ is identical to either $\alpha$, $-\alpha$, or $0$, if one listens at an output port where the vacuum state is present. To distinguish between these three events, a threshold $\epsilon>0$ is set and the detected value $x$ is declared as $\alpha$ if $x\geq\sqrt{n}\alpha-\epsilon$, as $-\alpha$ if $x\leq-\sqrt{n}\alpha+\epsilon$ and as $0$ if none of the above holds. The corresponding detection probabilities are equal to
    \begin{align}
        \mathbb P(\alpha|\beta)
            &=\frac{1}{2}\left\{ 1-\erf\left(\sqrt{2}(\beta-\epsilon)\right)\right\}\label{eq:detectalpha}  \\ 
        \mathbb P(0|\beta)
            &=\frac{1}{2}\left\{\erf\left(\sqrt{2}(\epsilon - \beta)\right) + \erf\left(\sqrt{2} (\epsilon + \beta)\right)\right\}\label{eq:detect0}  \\
        \mathbb P(-\alpha|\beta)
            &=\frac{1}{2}\left\{ 1 - \erf\left(\sqrt{2}(\epsilon+ \beta)\right)\right\}.
    \end{align}

    The conditional distribution \eqref{eqn:conditional-distribution} of the output port events given input $(\beta,k)$ and phase noise realization $\phi_1,\ldots,\phi_n$ can - due to the i.i.d. property of $\Phi_1,\ldots,\Phi_n$ - be calculated as 
    \begin{align}
        p_C(y|\beta)=\mathbb P(y|2^{-n/2}\beta\textstyle\sum_m e^{-\mathbbm{i}\phi_m})\\
        p_I(y|\beta)=\mathbb P(y|2^{-n/2}\beta\textstyle\sum_{m}{n}(-1)^m e^{-\mathbbm{i}\phi_m}).
    \end{align}
    For the simulations, we sampled $1000$ realizations of the phase noise values $(\phi_1,\ldots,\phi_n)$ to obtain an empirical approximation to \eqref{eqn:conditional-distribution}.
    \subsection{Proofs}

To prove the convergence of the joint detection receiver towards the expected value and the normal distribution we will have to first define the expected value and variance of the von Mises distribution.

\begin{lemma}[Raw Moments of the von Mises Distribution \label{lem:raw moments}]
Let $f(\theta;\beta,\kappa)= \frac{1}{2\pi I_0(\kappa)}e^{\kappa \cos(\theta-\beta)}$ be the von Mises distribution, with $0 \leq \beta < 2\pi$ and $\kappa \geq 0$ being parameters and $I_n$ the Bessel function of order $n$. Then the raw moments of this distribution are
\begin{align}
    m_n = E[X^n] = \frac{ I_n(\kappa)}{I_0(\kappa)} e^{\mathbbm{i}n\beta}
\end{align}
According to \cite{jammalamadaka2001topics} the central trigonometric moments are
\begin{align}
    \alpha_n^* = \frac{I_n(\kappa)}{I_0(\kappa)}\cdot \cos(n\beta),
\end{align}
whereas $\beta_n^* = 0$ due to the symmetry of the von Mises density.
\end{lemma}
For the proof see \cite[Chapter 2.2.4]{jammalamadaka2001topics}.

Now we can apply Hoeffding's inequality \cite[p.18, p.30]{vershynin_2018} to the Hadamard receiver with coherent states with phase noise.

\begin{proof}[Proof of Theorem \ref{thm:result}]
    Let $\phi_1,\ldots,\phi_n\in[0,2\pi)$ be a realization of the phase noise. Under this realization, the $k$-th sequence of signals $A_k := \alpha\cdot (H_{k,1}, \ldots, H_{k,n})$ is transformed to 
    \begin{align}
        \hat A_k := \alpha\cdot (H_{k,1}\exp{\mathbbm{i}\phi_1}, \ldots, H_{k,n}\exp{\mathbbm{i}\phi_n})
    \end{align} 
    and the $k$-th output of the Hadamard transform $H$ applied to $\hat A_k$ satisfies \begin{align}
        \tfrac{\alpha}{\sqrt{n}}\sum_m H_{k,m}H_{m,k}e^{\mathbbm{i}\Phi_m}=\tfrac{\alpha}{\sqrt{n}}\sum_me^{\mathbbm{i}\Phi_m}. 
    \end{align} 
    Since the right hand side is a sum of iid random variables we get from Hoeffding's inequality applied separately to the real and imaginary parts:
    \begin{align}
         \mathbb P\bigg(\bigg|\sum_m H_{k,m}H_{m,k}e^{\mathbbm{i}\Phi_m} - n\mathbb E_\mu [e^{\mathbbm{i}\Phi}]\bigg|\geq t\bigg)&\leq 4e^{-t^2/2n},
    \end{align}
with $\mathbb E_\mu [e^{\mathbbm{i}\Phi}]$ being the expectation value of $e^{\mathbbm{i}\Phi}$. 
Further for all $i\neq k$:
    \begin{align}
       \mathbb P\bigg(\bigg|\tfrac{\alpha}{\sqrt{n}}\sum_m H_{k',m}H_{m,k}e^{\mathbbm{i}\Phi_m}\bigg|\geq t\bigg)&\leq4e^{-t^2/n}
    \end{align}
\end{proof}

\begin{proof}[Proof of Theorem \ref{thm:capacity-bound}]
Assume $w(y^n|x,k)=g(y_k|x)\cdot\prod_{k'\neq k}d(y_{k'})$ for some conditional probability distribution $g$ and probability distribution $d$ on $\{-1,0,1\}$. To derive a lower bound on the capacity of such a system we let without loss of generality the input signals, consisting of BPSK symbols $x\in\{-1,1\}$ (where we set $\alpha=1$ without loss of generality) and input ports, are chosen uniformly at random. Then, the output distribution $q$ of the symbols $y^n\in\{-1,0,1\}^n$ can be written as
\begin{align}
    q(y^n)
        &=\sum_{x,k}\tfrac{1}{2n}g(y_k|x)\prod_{k\neq k'}d(y_{k'}).
\end{align} 
We can thus decompose this density as one truncated version on the set $A:=\{y^n:N(0|y^n)=n-1\}$ and $B:=A^\complement$ being the complement. Denote the respective re-scaled probabilities of $q$ restricted to $A$ or $B$ $q_A$ and $q_B$. Due to concavity of the entropy we have
\begin{align}
    H(q)
        &\geq q(A)H(q_A).
\end{align}
Then, for $y^n\in A$, without loss of generality $y_1=1$,
\begin{align}
    &q_A(y^n)=\tfrac{1}{2n}\bigg(g(y_1|y_1)\cdot d(0)^{n-1} + g(y_1|-y_1)\cdot d(0)^{n-1} +\nonumber\\
    &\qquad \sum_{k=2}^n\sum_{x}g(0|x)\cdot d(0)^{n-1}\bigg)\\
    &= \tfrac{d(0)^{n-1}g(1|1)}{2n}\bigg(1 + \tfrac{g(1|-1) + (n-1)\sum_{x}g(0|x)}{g(1|1)}\bigg). \label{eq:qa}
\end{align}
With $\pi$ denoting the uniform distribution on $A$ we get
\begin{align}
    q_A(y^n)=\lambda\pi(y^n)+(1-\lambda)\hat q_A(y^n)
\end{align}
with $\hat q_A$ defined in the obvious way from \eqref{eq:qa} and 
\begin{align}
    \lambda:=d(0)^{n-1}g(\alpha|\alpha).
\end{align}
Using concavity of the entropy and $H(\pi)=\log(2n)$ we get
\begin{align}
    H(q_A) &\geq
    \lambda(K+1) 
\end{align}
Now if the detector is well designed, then $d(0)\approx1$ and $g(y_1|y_1)\approx1$, so that we get $\lambda\approx1$ and therefore $H(q_A)\approx K+1$. 
Following our calculations, the capacity is lower bounded by
\begin{align}
    C(b,E,n)&\geq I(Y^n;AN)\\
    &\geq\lambda(K+1) - H(Y^n|AN),
\end{align}
where $A$ is a random variable describing the random choice of phases and $N$ a random variable describing the random choice of input port.
The distribution of $Y^n$ given $k$ and $x$ does not depend on the particular choice of $k,x$. Thus for the calculation of $H(Y^n|AN)$ it is sufficient to calculate the following:
\begin{align}
    -H&(Y^n|X=1,N=1) = \sum_{y^n}w(y^n|1,1)\log w(y^n|1,1)\\
        &= g(1|1)d(0)^{n-1}\log(g(1|1)d(0)^{n-1})  \nonumber\\
           &\qquad+ g(-1|1)d(0)^{n-1}\log(g(-1|1)d(0)^{n-1})\nonumber\\
            &\qquad -H(g(\cdot|1)) - g(\cdot|1)(n-1)H(d(\cdot))\\
        &= \lambda\log(\lambda) + 
            g(-1|1)d(0)^{n-1}\log(g(-1|1)d(0)^{n-1}) \nonumber\\
            &\qquad-H(g(\cdot|1)) - g(\cdot|1)(n-1)H(d(\cdot)).
\end{align}
Thus if $g(1|1)\geq1-\delta$ and $d(0)\geq1-\delta$ we get
\begin{align}
    -H(Y^n|AN) &\geq (1-\delta)\log(1-\delta) + \delta\log(\delta)\nonumber\\
    &\qquad - h(\delta) - \delta - \delta(n-1)(h(\delta) + \delta)\\
        &\geq -2\cdot h(\delta) - 2\delta - \delta(n-1)h(\delta). 
\end{align}
It thus follows with $\delta':=1-\delta$
\begin{align}
    C(b,E,n) &\geq \delta'(K+1) -2\cdot h(\delta) - 2\delta - \delta(n-1)h(\delta)
\end{align}
and for $\delta<\tfrac{1}{n-1}<\tfrac{1}{n}$
\begin{align}
    C(b,E,n)\geq (1-\delta)(K+1) -5 h(\delta).
\end{align}
\end{proof}

\end{section}



\begin{section}{Conclusions}
We have detailed the concept of a (quantum) joint detection receiver as a pre-processing step in classical receiver design. Improving upon earlier modelling steps, we have incorporated a scaling of phase noise with the baud-rate. Structural insights have been obtained by theoretical analysis. Our simulation results indicate that the theoretical performance predictions made in earlier works will persist in more realistic situations.

\section*{Acknowledgment}
Das Projekt/Forschungsvorhaben ist Teil der Initiative Munich Quantum Valley, die von der Bayerischen Staatsregierung aus Mitteln der Hightech Agenda Bayern Plus gefördert wird.
The project/research is part of the Munich Quantum Valley, which is supported by the Bavarian state government with funds from the Hightech Agenda Bayern Plus. Funding from the Federal Ministry of Education and Research of Germany, programme "Souver\"an. Digital. Vernetzt." joint project 6G-life, project identification number: 16KISK002 (ZA,JN), DFG Emmy-Noether program under grant number NO 1129/2-1 (JN) and support of the Munich Center for Quantum Science and Technology (MCQST) are acknowledged. Boulat A. Bash's work was supported in part by the National Science Foundation under Grant No. CCF-2006679.
\bibliographystyle{IEEEtran}
\bibliography{references.bib}

\begin{thebibliography}{10}
\providecommand{\url}[1]{#1}
\csname url@samestyle\endcsname
\providecommand{\newblock}{\relax}
\providecommand{\bibinfo}[2]{#2}
\providecommand{\BIBentrySTDinterwordspacing}{\spaceskip=0pt\relax}
\providecommand{\BIBentryALTinterwordstretchfactor}{4}
\providecommand{\BIBentryALTinterwordspacing}{\spaceskip=\fontdimen2\font plus
\BIBentryALTinterwordstretchfactor\fontdimen3\font minus
  \fontdimen4\font\relax}
\providecommand{\BIBforeignlanguage}[2]{{%
\expandafter\ifx\csname l@#1\endcsname\relax
\typeout{** WARNING: IEEEtran.bst: No hyphenation pattern has been}%
\typeout{** loaded for the language `#1'. Using the pattern for}%
\typeout{** the default language instead.}%
\else
\language=\csname l@#1\endcsname
\fi
#2}}
\providecommand{\BIBdecl}{\relax}
\BIBdecl

\bibitem{cisco}
\BIBentryALTinterwordspacing
Cisco, ``Cisco annual internet report,'' 2020. [Online]. Available:
  \url{https://www.cisco.com/c/en/us/solutions/executive-perspectives/annual-internet-report/index.html}
\BIBentrySTDinterwordspacing

\bibitem{6Grequirements}
M.~Ericson, S.~W\"anstedt, M.~Saimler, H.~Flinck, G.~Kunzmann, P.~Vlacheas,
  P.~Demestichas, D.~Rapone, A.~De~La~Oliva, C.~J. Bernardos, R.~Bassoli, F.~H.
  Fitzek, G.~Nardini, M.~Filippou, and M.~Mueck, ``Setting 6g architecture in
  motion – the hexa-x approach,'' in \emph{2022 Joint European Conference on
  Networks and Communications \& 6G Summit (EuCNC/6G Summit)}, 2022, pp.
  451--456.

\bibitem{6Gquantum}
R.~Bassoli, J.~N\"otzel, F.~H. Fitzek, H.~Boche, and N.~L. da~Fonseca, ``A
  novel architecture for future classical-quantum communication networks,''
  \emph{Wireless Communications and Mobile Computing}, vol. 2022, 2022.

\bibitem{Holevo1998c}
A.~S. Holevo, ``{Coding theorems for quantum communication channels},'' in
  \emph{IEEE Int. Symp. Inf. Theory - Proc.}, 1998, p.~84.

\bibitem{Giovannetti2015}
\BIBentryALTinterwordspacing
V.~Giovannetti, A.~S. Holevo, and R.~Garc{\'{i}}a-Patr{\'{o}}n, ``{A Solution
  of Gaussian Optimizer Conjecture for Quantum Channels},'' \emph{Commun. Math.
  Phys.}, vol. 334, no.~3, pp. 1553--1571, mar 2015. [Online]. Available:
  \url{http://link.springer.com/10.1007/s00220-014-2150-6}
\BIBentrySTDinterwordspacing

\bibitem{guha2011structured}
\BIBentryALTinterwordspacing
S.~Guha, ``Structured optical receivers to attain superadditive capacity and
  the holevo limit,'' \emph{Phys. Rev. Lett.}, vol. 106, p. 240502, Jun 2011.
  [Online]. Available:
  \url{https://link.aps.org/doi/10.1103/PhysRevLett.106.240502}
\BIBentrySTDinterwordspacing

\bibitem{noetzelrosati2022}
\BIBentryALTinterwordspacing
J.~N\"otzel and M.~Rosati, ``Operating fiber networks in the quantum limit,''
  2022. [Online]. Available: \url{https://arxiv.org/abs/2201.12397}
\BIBentrySTDinterwordspacing

\bibitem{highBaudrateComms}
S.-A. Li, H.~Huang, Z.~Pan, R.~Yin, Y.~Wang, Y.~Fang, Y.~Zhang, C.~Bao, Y.~Ren,
  Z.~Li, and Y.~Yue, ``Enabling technology in high-baud-rate coherent optical
  communication systems,'' \emph{IEEE Access}, vol.~8, pp. 111\,318--111\,329,
  2020.

\bibitem{cbandConverter}
G.~Ronniger, I.~Sackey, T.~Kernetzky, U.~Hoefler, C.~Mai, C.~Schubert,
  N.~Hanik, L.~Zimmermann, R.~Freund, and K.~Petermann, ``Efficient
  ultra-broadband c-to-o band converter based on multi-mode
  silicon-on-insulator waveguides,'' in \emph{2021 European Conference on
  Optical Communication (ECOC)}, 2021, pp. 1--4.

\bibitem{ludovicoWilde}
\BIBentryALTinterwordspacing
M.~W. L.~Lami, ``Exact solution for the quantum and private capacities of
  bosonic dephasing channels,'' 2022. [Online]. Available:
  \url{https://arxiv.org/abs/2205.05736}
\BIBentrySTDinterwordspacing

\bibitem{kunzParisBanaszekKerrMedium}
\BIBentryALTinterwordspacing
L.~Kunz, M.~G.~A. Paris, and K.~Banaszek, ``Noisy propagation of coherent
  states in a lossy kerr medium,'' \emph{J. Opt. Soc. Am. B}, vol.~35, no.~2,
  pp. 214--222, Feb 2018. [Online]. Available:
  \url{http://opg.optica.org/josab/abstract.cfm?URI=josab-35-2-214}
\BIBentrySTDinterwordspacing

\bibitem{kramerSystemModel}
B.~Goebel, R.-J. Essiambre, G.~Kramer, P.~J. Winzer, and N.~Hanik,
  ``Calculation of mutual information for partially coherent gaussian channels
  with applications to fiber optics,'' \emph{IEEE Transactions on Information
  Theory}, vol.~57, no.~9, pp. 5720--5736, 2011.

\bibitem{banaszekQuantumLimits}
K.~Banaszek, L.~Kunz, M.~Jachura, and M.~Jarzyna, ``Quantum limits in optical
  communications,'' \emph{Journal of Lightwave Technology}, vol.~38, no.~10,
  pp. 2741--2754, 2020.

\bibitem{rosati2016}
\BIBentryALTinterwordspacing
M.~Rosati, A.~Mari, and V.~Giovannetti, ``Multiphase hadamard receivers for
  classical communication on lossy bosonic channels,'' \emph{Physical Review
  A}, vol.~94, no.~6, Dec 2016. [Online]. Available:
  \url{http://dx.doi.org/10.1103/PhysRevA.94.062325}
\BIBentrySTDinterwordspacing

\bibitem{ourCode}
``Code is available at https://github.com/tqsd/hadamard/.''

\bibitem{colletLewisDiscriminatingWNDandVMD}
\BIBentryALTinterwordspacing
D.~Collett and T.~Lewis, ``Discriminating between the von mises and wrapped
  normal distributions,'' \emph{Australian Journal of Statistics}, vol.~23,
  no.~1, pp. 73--79, 1981. [Online]. Available:
  \url{https://onlinelibrary.wiley.com/doi/abs/10.1111/j.1467-842X.1981.tb00763.x}
\BIBentrySTDinterwordspacing

\bibitem{jammalamadaka2001topics}
S.~R. Jammalamadaka and A.~Sengupta, \emph{Topics in circular
  statistics}.\hskip 1em plus 0.5em minus 0.4em\relax world scientific, 2001,
  vol.~5.

\bibitem{vershynin_2018}
R.~Vershynin, \emph{High-dimensional probability: An introduction with
  applications in Data Science}.\hskip 1em plus 0.5em minus 0.4em\relax
  Cambridge University Press, 2018.

\end{thebibliography}
\end{section}
\end{document}